\documentclass[twocolumn,aps,showpacs,preprintnumbers, superscriptaddress, amsmath,amssymb]{revtex4}

\usepackage{stmaryrd} 
\usepackage{mathtools} 
\usepackage{stmaryrd} 

\newcommand*{\bra}[1]{\left\langle{#1}\right|}
\newcommand*{\ket}[1]{\left| #1 \right\rangle}

\newcommand*{\commut}[1]{[ #1 ]}

\newcommand*{\acommut}[1]{\{ #1 \}}
\newcommand*{\Acommut}[1]{\left\{ #1 \right\}}

\newcommand*{\poisson}[1]{\llbracket #1 \rrbracket}
\newcommand*{\Poisson}[1]{\left\llbracket #1 \right\rrbracket}

\usepackage{amsmath}
\usepackage{amsfonts}
\usepackage{amssymb}
\usepackage{graphicx}%

\newtheorem{theorem}{Theorem}
\newenvironment{proof}[1][Proof]{\noindent\textbf{#1.} }{\ \rule{0.5em}{0.5em}}

\begin{document}

\author{Renan Cabrera}
\email{rcabrera@princeton.edu}
\affiliation{Department of Chemistry, Princeton University, Princeton, NJ 08544, USA} 

\author{Denys I. Bondar}
\affiliation{Department of Chemistry, Princeton University, Princeton, NJ 08544, USA} 

\author{Herschel A. Rabitz}
\affiliation{Department of Chemistry, Princeton University, Princeton, NJ 08544, USA} 

\title{ Relativistic Wigner function and consistent classical limit for spin $1/2$ particles }

\date{\today}


\begin{abstract}
The relativistic Wigner function for spin $1/2$ particles 
is the subject of active research due to diverse applications. However, 
further progress is hindered by the fabulous complexity of the 
integro-differential equations of motion. We simplify these equations to partial differential equations of the Dirac type that are not only well suited for numerical computation, but also posses a well defined classical limit in a manifestly covariant form.
\end{abstract}


\maketitle

\emph{Introduction.} 
In non-relativistic quantum mechanics the phase space representation of the density 
operator is known as the Wigner quasi-probability distribution \cite{Wigner1932}, which is well suited for
the quantum-to-classical transition \cite{blokhintsev2010philosophy, Heller1976, zurek2001sub, bolivar2004quantum, zachos2005quantum, Kapral2006}.  The development of the relativistic Wigner formalism for spin $1/2$ particles, motivated by applications to quantum-plasma, had taken many decades \cite{hakim2011introduction} and reached its modern form in Refs. \cite{hakim1978covariant, hakim1982, vasak1987quantum}. A field-theoretic version of the Wigner function for spin $1/2$ particles was also formulated \cite{Bialynicki-birula1977, Bialynicki-Birula1991, shin1992wigner}.

Contrary to the non-relativistic case, the classical limit of the relativistic Wigner function is a subtle and complicated issue yet to be settled mainly due to the paramount complexity of the equations of motion. The existence of the classical limit in a non-covariant form was established in Refs. \cite{Bialynicki-birula1977, vasak1987quantum, Shin1993}. To restore covariance, Bolivar \cite{bolivar2001classical, bolivar2004quantum} presented a derivation based on rather exotic mathematical devices. However, the relation of Bolivar's classical limit to standard relativistic classical mechanics has not been established and remains a crucial open question. 

In the current Letter we drastically simplify the relativistic Wigner function's equations of motion making the classical limit a trivial exercise [Eq. (\ref{KND-gen})]. The resulting classical equations are manifestly covariant and shown to be equivalent to standard relativistic classical mechanics. Moreover, since the original integro-differential equations of motion are reduced to mere partial differential equations of the Dirac type [Eq. (\ref{WignerPsiEq})], our development can lead to more efficient numerical propagation of the relativistic Wigner function. 

A key element to accomplish this endeavor is to use the recently introduced 
Operational Dynamical Modeling (ODM) \cite{bondar2011hilbert} -- a universal and systematic framework for deriving equations of motion from the evolution of the dynamical average values. In Ref. \cite{bondar2011hilbert}, along with a number of other applications, we utilize this method to infer the Schr\"{o}dinger equation from the Ehrenfest theorems 
by assuming that the coordinate and momentum operators obey the canonical commutation relation. Otherwise if the coordinate and momentum commute, ODM leads to the Koopman-von Neumann mechanics \cite{Koopman1931,Neumann1932, Neumann1932a, mauro2003topics, Gozzi2010}, which is a Hilbert space formulation of non-relativistic classical mechanics with states being represented as complex valued wave functions and observables as commuting self-adjoint operators.
 In Secs. \ref{SalpeterSec}, \ref{Klein-Gordon-section}, and \ref{PauliEqSec}, we demonstrate the versatility of ODM in the case of relativistic classical and quantum mechanics by deriving the spineless Salpeter, Klein-Gordon, Vlasov, and Pauli equations. In the remaining part of the Letter, ODM is applied to relativistic spin $1/2$ particles.

\emph{Classical Spinorial Mechanics.}
Relativistic classical mechanics in manifestly covariant 
form with an electromagnetic interaction can be obtained from the extended Lagrangian
\begin{align}\label{ExtendedL}
 \mathcal{L} = m u^{\mu}u_{\mu} / 2 + e A^{\mu}u_{\mu} + mc^2 / 2,
\end{align} 
where the shell mass condition is incorporated as the integral of motion $u^{\mu}u_{\mu}=c^2$. The Euler-Lagrange equations
lead to the relativistic Newton equations 
\begin{align}
\label{classical-nrm}
m \frac{d u_{\mu}}{ds} = e F_{\mu\nu} u^{\nu}, \quad& u^{\mu}  =   \frac{d X^{\mu}}{ds},
\end{align}
which can be recast in the form of the following 
Hamilton equations (see Sec. \ref{standard-classical-dynamics})
\begin{align}
 \frac{d X^{\mu}}{ds} &= \frac{P^{\mu}-eA^{\mu}}{m}, \\
 \frac{d P_{\mu}}{ds} &=  \frac{e}{m} (\partial_{\mu} A_{\nu})( P^{\nu} -e A^{\nu}  ),
 \label{Newton4}
\end{align}
with $A_{\nu}$ as the four-vector potential. The components $X^{\mu}$ and $P_{\mu}$ define the extended phase space including the standard phase space for $\mu=1,2,3$ and the time-energy variables for $\mu=0$. 

Classical relativistic mechanics can also be expressed in spinorial form in terms of either the Spacetime Algebra \cite{hestenes1999new,doran2003geometric}
or the Algebra of Physical Space \cite{baylis1999electrodynamics,PhysRevA.45.4293,baylis2005quantum,PhysRevA.60.785,Baylis1996book,baylis2010quantum}. Inpired by these developments,  we adapt the Feynman slash notation. Let us express the proper velocity in Feynman slash notation 
\begin{align}
 \frac{d}{ds} X\!\!\!\!/ \equiv   u\!\!\!/ = u^{\mu}\gamma_{\mu} = 
  \frac{1}{4}Tr( u\!\!\!/ \gamma^{\mu}   ) \gamma_\mu,
\end{align}
where the gamma matrices obey the Clifford algebra in the Minkowski space
\begin{align}
  (\gamma^{\mu}\gamma^{\nu} +\gamma^{\mu}\gamma^{\nu}) = 2g^{\mu \nu} \mathbf{1}. 
\end{align} 
Given that $L$ is a Lorentz rotor, i.e., an element of the double cover of the proper Lorentz group,
 the proper velocity transforms as 
$
   u\!\!\!/ \rightarrow u^\prime\!\!\!\!\!/ = L u\!\!\!/ L^{-1},
$
where the inverse of the Lorentz rotor obeys $L^{-1} = \gamma^{0}L^{\dagger}\gamma^0$. 
Any proper velocity is obtained through a Lorentz transformation of the particle at rest with proper velocity $   u\!\!\!/_{\text{rest}} = c \gamma^0$.
Hence, all the information stored in $   u\!\!\!/$ can be included in the Lorentz rotor $L$.
The classical column spinor $\Psi$ is defined as the leftmost columns of $L$ as
 $ \Psi =   L |_{\text{leftmost column}}$ and the following relation
holds \cite{hestenes1973local,hestenes1975observables}
\begin{align}
   \frac{d X^{\mu}}{ds} = \Psi^\dagger c \gamma^0\gamma^{\mu}\Psi.
\label{prop-vel-psi}
\end{align}
Note that the column spinor $\Psi$ contains all the information stored in the Lorentz rotor 
$L$ despite the fact that $\Psi$ is only a single column of $L$ \cite{hestenes1973local,hestenes1975observables}. 
It is possible to rewrite the Dirac equation in terms of Lorentz rotors solely and employ such a representation to obtain solutions
 (see, e.g., Sec. \ref{Dirac-free-particle}).
Using the properties of the classical spinor $\Psi$, the relativistic Newton equation (\ref{Newton4}) is expressed  as
\begin{align}
      m \frac{d}{ds} u_{\mu} =   \Psi^{\dagger} c \gamma^0 \gamma^{\nu} e F_{\mu \nu} \Psi, \quad
 \frac{d}{d s} P_{\mu} =  \Psi^{\dagger}  c \gamma^0 e \partial_{\mu} A\!\!\!/ \Psi,
\label{Lorentz-force4}
\end{align}
which along with (\ref{prop-vel-psi}) define a system of differential equations that 
can be solved once initial conditions are supplied \cite{hestenes1999new,PhysRevA.60.785}.  Classical spinors are further discussed in Sec.  \ref{classical-spinor-appendix}.

\emph{The Dirac equation.} 
Let us derive the Diract equation using ODM. According to Ref. \cite{bondar2011hilbert}, in order to construct a system's dynamical model, ODM requires the following three elements: i) the evolution of the average values drawn in the form of Ehrenfest-like theorems, ii) the definition of the observables' average, and iii) the algebra of the observables.

The first ODM element is given by
\begin{align}
     \frac{d }{ds} \overline{X^{\mu}} =  \overline{ \Psi^\dagger c \gamma^0\gamma^{\mu}\Psi }  ,\qquad
      \frac{d}{d s} \overline{ P_{\mu} } =  \overline{ \Psi^{\dagger}  c \gamma^0 e \partial_{\mu} A\!\!\!/ \Psi }.
  \label{first-ODM-spinorial}
\end{align}
Even though these equation resemble Eqs. (\ref{prop-vel-psi}) and (\ref{Lorentz-force4}), 
they are in fact more general because the averages are performed on an ensemble of
particles that need not have the notion of  well defined trajectories. The second ODM element specifies averaging
in terms of the spinorial Hilbert space where the observables are replaced by the corresponding 
operators
\begin{align}
\label{rel-Ehrenfest-1}
 \frac{d}{ds} \langle \psi | \hat{x}^{\mu} |\psi \rangle & = 
  \langle \psi |  c \gamma^0 \gamma^{\mu}  |\psi \rangle , \\
 \frac{d }{ds} \langle \psi | \hat{p}_{\mu} |\psi \rangle & 
  =  \langle \psi | c e \partial_{\mu} \hat{A}_{\nu} \gamma^0 \gamma^{\nu}  |\psi \rangle, 
\label{rel-Ehrenfest-2}
\end{align}
which are  relativistic generalizations of the Ehrenfest theorems.
The Dirac generator of motion $\boldsymbol{D}$ defined as
\begin{align}
 \hbar \frac{d | \psi \rangle}{ds}  = i \boldsymbol{D} |\psi\rangle, 
\end{align}
applied to Eqs. (\ref{rel-Ehrenfest-1}) and (\ref{rel-Ehrenfest-2}) leads to
\begin{align}
 \langle \psi | i[ \hat{x}^{\mu} , \boldsymbol{D} ] / \hbar | \psi \rangle 
&= c \gamma^0 \gamma^{\mu} , \\
 \langle \psi | i[ \hat{p}_{\mu} , \boldsymbol{D}] / \hbar | \psi \rangle 
  &= c  e  \langle \psi |  \partial_{\mu} \hat{A}_{\nu} \gamma^0 \gamma^{\nu}  | \psi \rangle. 
\end{align}
The expectation values can be removed considering that these identities should be 
valid for all possible initial conditions
\begin{align}
 i[ \hat{x}^{\mu} , \boldsymbol{D} ]  / \hbar = c \gamma^0 \gamma^{\mu} , \quad
 i[ \hat{p}_{\mu} , \boldsymbol{D}]  / \hbar = c  e \partial_{\mu} \hat{A}_{\nu} \gamma^0 \gamma^{\nu}.
\label{dirac-comm-eqs}
\end{align}
The kinematic framework of quantum mechanics is implemented
as the third ODM element by demanding $\hat{x}^\mu$ and $\hat{p}^{\mu}$ to satisfy the canonical commutation 
relations  $[ \hat{x}^{\mu} , \hat{p}_{\nu} ] = - i \delta^{\mu}_{\,\,\,\,\nu} \hbar$. Theorem 1 in Sec. \ref{Sec_Weyl_calculus} is used to convert Eq. (\ref{dirac-comm-eqs}) to the system 
of partial differential equations
\begin{align}
 \frac{\partial}{\partial \hat{p}_{\mu}}  \boldsymbol{D}  = c \gamma^0 \gamma^{\mu} ,  \qquad
 - \frac{\partial}{\partial \hat{x}^{\mu}} \boldsymbol{D} = c e \partial_{\mu} \hat{A}_{\nu} \gamma^0 \gamma^{\nu}, 
\end{align}
whose solution is the self-adjoint Dirac generator
\begin{align}
   \boldsymbol{D} = c \gamma^0 \gamma^{\mu}( \hat{p}_{\mu} - e \hat{A}_{\mu})  -  \gamma^0 m c^2,
\label{Dirac-gens}
\end{align}
where the integration constant $\gamma^0 m c^2$ is specified by demanding consistency 
with the Klein-Gordon equation (see further discussions in Sec. \ref{Klein-Gordon-section}). 
The Dirac generator is independent of the parameter $s$ implying existence of the integral of motion -- the shell mass condition that reads $\boldsymbol{D}=0$. Finally, the Dirac equation is $\boldsymbol{D} \psi = 0$.

The Dirac equation allows for existence of
negative energy particles responsible for the  \emph{zitterbewegung} effect and may ultimately lead to the break down of Eq. (\ref{first-ODM-spinorial}). The physical condition for the absence of negative energy particles requires wave packets to be no more localized than the Compton wavelength $\Delta x \ll \frac{\hbar}{mc}$ \cite{greiner2000relativisticIB}, 
while maintaining variations of the potential not larger than the energy necessary for pair creation 
in a region smaller than the Compton wavelength $ \frac{\Delta V}{\Delta x}  \ll \frac{2mc^2}{\hbar/(mc)} $ \cite{greinerReinhardtBookPart}. 

\emph{Relativistic Wigner function.} The Dirac theory for a spin $1/2$ particle 
was expressed in terms of the extended phase space 
\cite{hakim1978covariant,hakim1982,hakim2011introduction,vasak1987quantum}
and plays a central role in relativistic quantum transport.
Following the ODM derivation of the non-relativistic phase space representation \cite{bondar2011hilbert}, we reformulate the relativistic Wigner function in
Blokhintsev's  double configuration space \cite{Blokhintsev1977,Blokhintsev1941,Blokhintsev1940a,Blokhintsev1940,blokhintsev2010philosophy} leading to the significant simplification of equations of motion.

The first ODM element is again given by Eqs. (\ref{rel-Ehrenfest-1}) and (\ref{rel-Ehrenfest-2}).
Reproducing the steps leading to the Dirac equation, we first define the Wigner generator $\boldsymbol{W}$ as 
\begin{align}
 \hbar \frac{d | \psi \rangle}{ds}  = i \boldsymbol{W} |\psi\rangle, 
\end{align}
and obtain the equations for $\boldsymbol{W}$
\begin{align}
 \label{W-comm-eq}
 i[ \hat{x}^{\mu} , \boldsymbol{W} ] / \hbar= c \gamma^0 \gamma^{\mu} , \quad
 i[ \hat{p}_{\mu} , \boldsymbol{W}] / \hbar = c  e \partial_{\mu} \hat{A}_{\nu} \gamma^0 \gamma^{\nu}. 
\end{align}
In contrast to the Dirac equation's derivation, the quantum algebra of operators 
is extended with additional operators $\hat{\lambda}_{\mu}$ and $\hat{\theta}_{\mu}$ such that
\begin{align}
  {[} \hat{x}^{\mu} , \hat{p}_{\nu}  {]} &= -i \hbar \kappa \delta^{\mu}_{\,\,\,\,\nu}, \qquad & 
  {[} \hat{x}^{\mu} , \hat{\lambda}_{\nu}  {]} &= -i \delta^{\mu}_{\,\,\,\,\nu} , \nonumber \\
  {[}  \hat{p}_{\mu} , \hat{\theta}^{\nu}  {]} &= -i \delta^{\nu}_{\,\,\,\,\mu}, \qquad &
  {[}\hat{\lambda}_{\mu} , \hat{\theta}^{\nu}{]} &= 0,  \label{qa-1}
\end{align}
with all the other commutators vanishing and  $\kappa$ being  a real parameter
that measures the degree of commutativity/quantumness, such that quantum mechanics 
is recovered at $\kappa \rightarrow 1$,  while $\kappa \rightarrow 0$ corresponds to the classical limit. This type of algebra was employed in Ref. \cite{bondar2011hilbert} 
to deduce the Koopman-von Neumann equation.

Another convenient set of operators, to be called the \emph{mirror quantum algebra}, is defined as 
\begin{align}
  \hat{x}^{\prime \mu} &=   \hat{x}^{\mu} + \hbar \kappa \hat{\theta}^{\mu},  \qquad&     \hat{\lambda}^{\prime}_{\mu} &=  - \hat{\lambda}_{\mu},   \nonumber \\
  \hat{p}^{\prime}_{\nu} &=  \hat{p}_{\nu} - \hbar \kappa \hat{\lambda}_{\nu}, \qquad&
  \hat{\theta}^{\prime\mu} &=  - \hat{\theta}^{\mu},
\end{align}
and obeys the canonical quantum commutations
relations (\ref{qa-1}) with the opposite sign 
\begin{align}
\label{qa-mirrow-1}
  {[} \hat{x}^{\prime \mu} , \hat{p}^{\prime}_{\nu}  {]} &=  i \hbar \kappa \delta^{\mu}_{\,\,\,\,\nu} , \qquad&
   {[}  \hat{p}^{\prime}_{\mu} , \hat{\theta}^{\prime \nu}  {]} &= i \delta^{\nu}_{\,\,\,\,\mu}, \nonumber \\
  {[} \hat{x}^{\prime \mu} , \hat{\lambda}^{\prime}_{\nu}  {]} &= i \delta^{\mu}_{\,\,\,\,\nu} , \qquad&
 {[} \hat{\lambda}^{\prime}_{\mu} , \hat{\theta}^{\prime \nu} {]} &= 0.
\end{align}
Moreover, any function  $ f( \hat{x}^{\prime} , \hat{p}^{\prime}   )$ acts as a constant  
with  respect the quantum operators $\hat{x}^{\mu}$ and $\hat{p}_{\mu}$ because
\begin{align}
 \commut{ \hat{x}^{\prime \mu}  , \hat{x}^{ \nu}  } &= 0 ,\qquad&
  \commut{ \hat{x}^{\prime \mu} , \hat{p}_{\nu}  } &= 0 ,\nonumber \\
    \commut{ \hat{p}^{\prime}_{\mu}  , \hat{x}^{\nu}  } &= 0 ,\qquad& 
  \commut{ \hat{p}^{\prime}_{\mu}  , \hat{p}_{\nu}  } &= 0.
\end{align} 
Similar properties also applies to  $ f( \hat{x} , \hat{p}   )$, which acts as a constant  with  
respect to the mirror quantum operators $\hat{x}^{\prime \mu}$ and $\hat{p}^{\prime}_{\mu}$.

According to Theorem 1 in Sec.  \ref{Sec_Weyl_calculus}, 
the commutator equations (\ref{W-comm-eq}) lead to the system 
of differential equations 
\begin{eqnarray}
\frac{i}{\hbar}\left( 
 -i \hbar \kappa  \frac{\partial \boldsymbol{W}}{\partial \hat{p}_{\mu} }
 -i \frac{\partial \boldsymbol{W} }{\partial \hat{\lambda}_{\mu} }
\right) &=& c \gamma^0\gamma^{\mu}  , \\
\frac{i}{\hbar}\left( 
 i \hbar \kappa \frac{\partial \boldsymbol{W} }{\partial \hat{x}^{\mu} } 
  + i  \frac{\partial \boldsymbol{W} }{\partial \hat{\theta}^{\mu} }   \right) &=& c e \partial_{\mu} \hat{A}_{\nu}\gamma^0\gamma^{\nu},
\end{eqnarray}
with the following solution 
\begin{align}
  \boldsymbol{ W } = \frac{\gamma^0}{\kappa}
   \left( c \gamma^{\mu} \left[ \hat{p}_{\mu} - e \hat{A}_{\mu} 
+ f_{\mu}( \hat{x}^{\prime} , \hat{p}^{\prime} )
 \right]  -  \kappa  m c^2 \right) 
    ,
\label{unified-Dirac}
\end{align}
where $ f_{\mu}( \hat{x}^{\prime}, \hat{p}^{\prime})$ is an arbitrary function that can be absorbed by a gauge transformation. 
A particularly useful choice is $ f_{\mu} = - (\hat{p}_{\mu} - \hbar \kappa \hat{\lambda}_{\mu} )- e \hat{A}_{\mu}( \hat{x} + \hbar \kappa \hat{\theta}  ) $, which leads to
\begin{align}
  \boldsymbol{ W } = \gamma^0
   \left( \frac{c \gamma^{\mu}}{\kappa}[ \hbar \kappa \hat{\lambda}_{\mu}  - e \hat{A}_{\mu}(\hat{x}) 
+e \hat{A}_{\mu}( \hat{x} + \hbar \kappa \hat{\theta}  )
 ]  -   m c^2 \right) 
\label{unified-Dirac-A-A}.
\end{align}
The Wigner generator $\boldsymbol{W}$ can be rewritten by defining an independent classical 
algebra of operators $ \hat{X}^{\mu} , \hat{P}_{\nu} ,  \hat{\Lambda}^{\mu} , \hat{\Theta}_{\nu} $ obeying
\begin{align}
 \label{commutation-classical1}
 [ \hat{X}^{\mu} , \hat{P}_{\nu} ] &= 0, \qquad&
 {[} \hat{X}^{\mu} , \hat{\Lambda}_{\nu}  {]} &= -i \delta^{\mu}_{\,\,\,\,\nu}, \\
 {[}  \hat{P}_{\mu},   \hat{\Theta}^{\nu}    {]} &= -i \delta^{\nu}_{\,\,\,\,\mu}, \qquad&
 {[}\hat{\Lambda}_{\mu} , \hat{\Theta}^{\nu}{]} &= 0  ,
\end{align}
such that the extended quantum algebra can be consistently expressed  as
\begin{align}
  \label{quantum-classical-op-1}
  \hat{x}^{\mu} &= \hat{X}^{\mu} - \frac{\hbar\kappa}{2} \hat{\Theta}^{\mu} ,\qquad&
  \hat{\lambda}_{\mu} &= \hat{\Lambda}_{\mu} ,\\\
  \hat{p}_{\mu} &= \hat{P}_{\mu} + \frac{\hbar\kappa}{2} \hat{\Lambda}_{\mu}  ,\qquad&
  \hat{\theta}^{\mu} &= \hat{\Theta}^{\mu}. 
\end{align}
The Wigner generator $\boldsymbol{W}$ written in terms of the classical operators is
\begin{align}
  \boldsymbol{ W } = 
    c\gamma^0 \gamma^{\mu} \left[ \hbar  \hat{\Lambda}_{\mu}
 + \frac{e}{\kappa} \left(  \hat{A}_{\mu}^{+} -  \hat{A}_{\mu}^{-} \right) \right]
   -   \gamma^0 m c^2  
   ,
\label{unified-Dirac-AA}
\end{align}
with $A_{\mu}^{-}=   A_{\mu}( \hat{X} - \frac{\hbar \kappa}{2} \hat{\Theta}  )  $ and $A_{\mu}^{+}=   A_{\mu}( \hat{X} + \frac{\hbar \kappa}{2} \hat{\Theta}  )  $. 
This form of the Wigner generator has  
the explicit and well defined classical limit 
\begin{align}
\label{KND-gen}
& \lim_{\kappa \rightarrow 0} \boldsymbol{W} = \hbar \boldsymbol{K} , \nonumber\\
 & \boldsymbol{K} 
  = c \gamma^0\gamma^{\mu} \left( \hat{\Lambda}_{\mu} + e \frac{\partial}{\partial \hat{X}^{\nu}}
 \hat{A}_{\mu} \hat{\Theta}^{\nu}  \right) - \gamma^0 \frac{mc^2}{\hbar},
\end{align}
where $\boldsymbol{K}$ shall be referred to as the Koopman-von Neumann-Dirac (KvND) generator, overlooked in the literature, which is a relativistic generalization of the non-relativistic Koopman-von Neumann generator \cite{Koopman1931,Neumann1932, Neumann1932a, mauro2003topics, Carta2006, Gozzi2010, Gozzi2011, Cattaruzza2011}. Bogdanov \cite{bogdanov} reported a special case of Eq. (\ref{KND-gen}). 

A gauge independent classical generator can be obtained from Eq. (\ref{KND-gen}) using 
kinetic instead of the canonical momentum. The kinetic 
momentum equation of motion in the Heisemberg picture reads
\begin{align}
m \frac{d \hat{u}_{\mu} }{ds} = i \commut{ \hat{p}_{\mu} - e \hat{A}_{\mu} , \boldsymbol{K} }
=  \gamma^0\gamma^{\nu} e \hat{F}_{\mu \nu}, 
\end{align}
where $\hat{F}_{\mu\nu}$ is the electromagnetic field as a function of the position operator. 
A similar procedure can be carried out to find the variation of the position operator $\hat{x}^{\mu}$ 
\begin{align}
 \frac{d \hat{x}^{\mu}}{ds} = c \gamma^0\gamma^{\mu},
\end{align}
such that one recovers the standard formulation of classical relativistic mechanics (\ref{classical-nrm}). It is also
possible to define a gauge independent classical generator in terms 
of the position and proper velocity 
\begin{align}
 \boldsymbol{V} = i \gamma^0 \gamma^{\mu}\left( \frac{\partial}{\partial X^{\mu}} + e F_{\nu \mu} \frac{\partial}{\partial u_{\nu}} \right).
\end{align}

The Wigner generator $\boldsymbol{W}$ can be explicitly represented through the extended phase 
representation resulting in an intricate system of integro-differential equations \cite{vasak1987quantum,Bialynicki-Birula1991}. 
A better alternative, leading to simpler analytical analyses as well as efficient numerical methods, is to use Blokhintsev's $X\Theta$-representation 
\begin{align}
  \hat{X}^{\mu} =  X^{\mu}, \quad
  \hat{\Lambda}_{\mu}  = i \frac{\partial \,\,\,}{ \partial X^{\mu}}, \quad
   \hat{P}_{\mu} =   -i \frac{\partial \,\,\,}{ \partial \Theta^{\mu}}, \quad
  \hat{\Theta}^{\mu}  = \Theta^{\mu},
 \label{PTheta-Representation}
\end{align}
which is equivalent to doubling the configuration space with the additional variable $\Theta$.
The spinorial equation of motion in  the $X\Theta$ representation is 
a system of partial differential equations 
\begin{align}\gamma^0
  \left(  c \gamma^{\mu}\left[ i \hbar \frac{\partial}{\partial X^{\mu}} 
 + \frac{e}{\kappa}(  A_{\mu}^{+} -  A_{\mu}^{-} ) \right]
   -   m c^2  \right)\psi(X,\Theta) = 0, \label{WignerPsiEq}
\end{align}
where $\psi(X,\Theta)$ is the spinorial relativistic Wigner function encapsulating all the  dynamical information.

In relativistic quantum mechanics the Dirac current is defined as $J^{\mu} = \psi^{\dagger} \gamma^{\mu}\gamma^0 \psi  $, 
which is written in slash notation as
$J\!\!\!/ = \psi \psi^{\dagger} \gamma^0$ \cite{hestenes1973local,hestenes1975observables}. Therefore, the current 
in the $X\Theta$-representation reads
\begin{align}
  J_{X\Theta}(X,\Theta) = \psi(X,\Theta) \psi^{\dagger}(X,\Theta)\gamma^0.
\end{align}
Considering that $\Theta$ is the conjugate variable of $P$ [Eq. (\ref{PTheta-Representation})],
the current in the $XP$-representation becomes
\begin{align}
 J_{XP}(X,P) =  \int d^4 \Theta \, \psi(X,\Theta )  
 \psi^\dagger(X, \Theta ) \gamma^0 \exp( - i P_{\nu} \Theta^{\nu}  ), 
\label{J-Wigner}
\end{align}
which is the relativistic Wigner function in the extended phase space for spin $1/2$ particles. Contrary to the non-relativistic case when the Wigner function can be interpreted as a probability density (though not positively defined) as well as a Koopman-von Neuman-like wave function \cite{bondar2012wigner},
the relativistic Wigner function is a current. 

Note Eq. (\ref{J-Wigner}) is not gauge independent
because the canonical momentum $P$ is the sum of both the kinetic momentum 
and the gauge dependent vector potential. Nevertheless, there are gauge independent representations solely
in terms of the kinetic momentum \cite{stratonovich1956gauge,serimaa1986gauge,vasak1987quantum}. 

\emph{Outlook}. Having simplified the relativistic Wigner function's equations of motion for spin $1/2$ particles, we foremost settle down the issue of the classical limit. More important, since the simplified dynamical equation [Eq. (\ref{WignerPsiEq})] is of the Dirac type, possessing efficient numerical schemes \cite{Fillion-Gourdeau2011}, the current work may open new horizons for further advancement of numerical algorithms. Besides applications in relativistic quantum transport theory \cite{hakim2011introduction}, the presented Wigner function construction can be readily extended to non-abelian relativistic statistical mechanics 
 (see, e.g., Ref. \cite{ochs1998wigner}) as well as to solid state physics where 
the effective relativistic Wigner function is needed \cite{morandi2011wigner}. 

\emph{Acknowledgments.} The authors are financially supported by DARPA, NSF, and ARO.

\begin{widetext}
\begin{center}
{\bf Supplementary Material for: ``Relativistic Wigner function and consistent classical limit for spin $1/2$ particles''}
\end{center}

In the main text we develop relativistic quantum mechanics through Operational Dynamicla Modelling (ODM) \cite{bondar2011hilbert}. The purpose of this supplementary material is to provide background information and additional examples of the
use of ODM.  

\tableofcontents

\section{Classical Dynamics in the Extended Phase Space}
\label{standard-classical-dynamics}
The Lagrangian of a relativistic particle in an electromagnetic field without the shell mass constraint is 
\begin{align}
 \label{extended-Lagrangian}
 \mathcal{L} = \frac{m}{2} u^{\mu}u_{\mu} + e A^{\mu}u_{\mu} + \frac{mc^2}{2},
\end{align}
where $u^{\mu} = \frac{d x^{\mu}}{ds}$ are the components of the proper velocity (four-velocity) defined
in the Minkowski space with a metric with diagonal metric elements $\{1,-1,-1,-1\}$. 
More precisely, this is the time-extended form  of the Lagrangian of a relativistic particle 
with electromagnetic interaction \cite{walter1998classical}. In this formalism the shell constrain
 \begin{align}
u^{\mu}u_{\mu} = c^2
\label{uu}
\end{align} 
is implemented as an integral of motion in order to provide physical solutions. 
The canonical momentum is denoted with covariant indices as
\begin{align}
 P_{\mu} \equiv \frac{\partial \mathcal{L}}{\partial u^{\mu}} = m u_{\mu} + e A_{\mu},
\end{align} 
where the physical momentum and four-vector potential are $P^{\mu} = (E/c, \boldsymbol{P})$,  $A^{\mu} = ( \phi/c , \boldsymbol{A}  )$, 
where bold symbols stand for vectors in the standard three-dimensional Euclidean space.

The time-extended Lagrangian leads to the \emph{extended phase space} Hamiltonian
\begin{align}
  \mathcal{H} \equiv P_{\mu} u^{\mu} -\mathcal{L} = \frac{(P^{\mu}-e A^{\mu})( P_{\mu} -eA_{\mu} ) }{2m} - \frac{mc^2}{2}.
\label{te-H}
\end{align}
This Hamiltonian does not explicitly depend on the parameter $s$, 
therefore, it is an invariant  dynamical integral of motion. The particular integration condition 
that leads to physical solutions according to the shell mass (\ref{uu}) is
 \begin{align}
\mathcal{H}=0.
\label{H-zero}
\end{align}

The equations of motion also can be obtained in terms of the extended phase space Poisson 
brackets defined as
\begin{align}
   \poisson{F,G} \equiv 
 \frac{\partial F}{\partial X^{\mu}} \frac{\partial G}{\partial P_{\mu}} 
 - \frac{\partial G}{\partial X^{\mu}} \frac{\partial F}{\partial P_{\mu}}  
\end{align}
such that 
\begin{align}
 \frac{dF}{ds} = \poisson{F, \mathcal{H} }.
\end{align}
The four-velocity components are recovered as
\begin{align}
  \frac{d X^{\mu}}{ds} &= \frac{\partial \mathcal{H}}{ \partial P_{\mu}} = \frac{P^{\mu}-eA^{\mu}}{m}
\label{fourvel}
\end{align}
and the canonical four-force equation is 
\begin{align}
\frac{d P_{\mu}}{ds} &= -\frac{\partial \mathcal{H}}{ \partial X^{\mu}} 
  = \frac{e}{m} (\partial_{\mu} A_{\nu})( P^{\nu} -e A^{\nu}  ), \nonumber \\
\frac{d P_{\mu}}{ds} &= e(\partial_{\mu} A_{\nu}) u^{\nu}.
\label{force4}
\end{align}
The left side of this equation can be expressed as
\begin{align}
 \frac{d P_{\mu}}{ds} = m \frac{d u_{\mu}}{ds} + e\frac{d A_{\mu}}{ds} = 
m \frac{d u_{\mu}}{ds} + e\frac{\partial A_{\mu}}{\partial X^{\nu}} \frac{d X^{\nu}}{d s} = 
 m \frac{d u_{\mu}}{ds} + e\frac{\partial A_{\mu}}{\partial X^{\nu}} u^{\nu}
\label{dp-canonical},
\end{align}
which leads to the Lorentz force in covariant form upon 
substitution in Eq. (\ref{force4})
\begin{align}
m \frac{d u_{\mu}}{ds} &= e(  \partial_{\mu} A_{\nu} - \partial_{\nu}A_{\mu} )  u^{\nu}.
\end{align}
The  latter is expressed in terms of the Faraday electromagnetic tensor
 $F_{\mu \nu} = \partial_{\mu} A_{\nu} - \partial_{\nu}A_{\mu}  $ as
\begin{align}
m \frac{d u_{\mu}}{ds} &= e F_{\mu\nu} u^{\nu}.
\end{align}

The action can be calculated  using the time-extended Lagrangian $\mathcal{L}$ (\ref{extended-Lagrangian})
or the corresponding standard Lagrangian $L$, where the time is considered as a parameter. 
Both methods must lead to the same result, implying
\begin{eqnarray}
 \int \mathcal{L} ds &=& \int L dt ,\\
 \int \mathcal{L} \frac{ds}{dt} dt  &=& \int L dt,
\end{eqnarray}
thus, up to some exact differential one has the following equality
\begin{equation}
 L = \mathcal{L} \frac{ds}{dt}.
\end{equation}

\section{Noncommutative Analysis: The Weyl Calculus}
\label{Sec_Weyl_calculus}
Noncommutative analysis \cite{Weyl1950, Maslov1976,Nazaikinskii1996} is a broad and active field of mathematics with a number of important applications. This branch of analysis aims at identifying functions of noncommutative variables and specifying operations with such objects. There are many ways of introducing functions of operators; however, the choice of a particular definition is a matter of convenience \cite{Nazaikinskii1992}. 

To make the paper self-consistent, we shall review basic results from the Weyl calculus, which is a popular version of 
noncommuting analysis. Theorem \ref{Th_Weyl_commutator_theorem} plays a crucial role in the current paper. Even though we prove this 
result within the Weyl calculus, it is valid in more general settings (see, e.g., Ref. \cite{Maslov1976} and 
page 63 of Ref. \cite{Nazaikinskii1996}).

The starting point is the well known fact that Fourier transforming back and forth  does not change a 
sufficiently smooth function of $n$-arguments,
\begin{align}\label{Fourier_identity}
	f(\lambda_1, \ldots, \lambda_n) = \frac{1}{(2\pi)^n} \int \prod_{l=1}^n d\xi_l d\eta_l 
		\exp\left[ i \sum_{q=1}^n \eta_q(\lambda_q - \xi_q) \right] f(\xi_1, \ldots, \xi_n).
\end{align}
Following this observation, we define the function of noncommuting operators within the Weyl calculus as 
\begin{align}
\label{Weyl_main_definition}
	f(\hat{A}_1, \ldots, \hat{A}_n) = \frac{1}{(2\pi)^n} \int \prod_{l=1}^n d\xi_l d\eta_l
		\exp\left[ i \sum_{q=1}^n \eta_q(\hat{A}_q - \xi_q) \right] f(\xi_1, \ldots, \xi_n),
\end{align}
where the exponential of an operator is specified by the Taylor expansion,
\begin{align}\label{OperatorExp}
	\exp(\hat{A}) = \sum_{k=0}^{\infty} \frac{\hat{A}^k}{ k! }.
\end{align}
The identity
\begin{align}\nonumber
	f^{\dagger} (\hat{A}_1, \ldots, \hat{A}_n) = f(\hat{A}_1^{\dagger}, \ldots, \hat{A}_n^{\dagger})
\end{align}
implies that the function of self-adjoint operators (\ref{Weyl_main_definition}) is itself a self-adjoint operator. 
Moreover, one may demonstrate that  
\begin{align}\label{Weyl_derivative_of_function}
	f'_{\hat{A}_k} (\hat{A}_1, \ldots, \hat{A}_n) & = \lim_{\epsilon\to 0} \frac 1{\epsilon} 
		\left[ f(\hat{A}_1, \ldots, \hat{A}_k + \epsilon, \ldots, \hat{A}_n)  - f(\hat{A}_1, \ldots, \hat{A}_k, \ldots, \hat{A}_n)\right]\nonumber\\
		&= \frac{1}{(2\pi)^n} \int \prod_{l=1}^n d\xi_l d\eta_l \, i\eta_k 
			\exp\left[ i \sum_{q=1}^n \eta_q(\hat{A}_q - \xi_q) \right] f(\xi_1, \ldots, \xi_n) \nonumber\\
		&= \frac{1}{(2\pi)^n} \int \prod_{l=1}^n d\xi_l d\eta_l
			f(\xi_1, \ldots, \xi_n) \left( -\frac{\partial}{\partial \xi_k} \right) \exp\left[ i \sum_{q=1}^n \eta_q(\hat{A}_q - \xi_q) \right]  \nonumber\\
		&= \frac{1}{(2\pi)^n} \int \prod_{l=1}^n d\xi_l d\eta_l
			\exp\left[ i \sum_{q=1}^n \eta_q(\hat{A}_q - \xi_q) \right] f'_{\xi_k} (\xi_1, \ldots, \xi_n).
\end{align}
Equation (\ref{Weyl_main_definition}) defines a one-to-one mapping between a function $f(\xi_1, \ldots, \xi_n)$ and a linear operator $f(\hat{A}_1, \ldots, \hat{A}_n)$. By the same token, Eq. (\ref{Weyl_derivative_of_function}) establishes a one-to-one mapping between the derivative of a function and the derivative of a linear operator.

The following theorem is of fundamental importance:
\begin{theorem}\label{Th_Weyl_commutator_theorem}
	Let $\hat{A}_1, \ldots, \hat{A}_n$ be some operators and $\hat{C}_k = \commut{\hat{A}_k, \hat{B}}$, $k=1,\ldots,n$. If 
	$\commut{ \hat{A}_k, \hat{C}_l } = \commut{  \hat{B}, \hat{C}_k } = 0$, $k,l=1,\ldots,n$, then
	\begin{align}\label{Weyl_MainTh_equality}
		\commut{ f(\hat{A}_1,\ldots, \hat{A}_n), \hat{B} } = \sum_{k=1}^n \commut{\hat{A}_k, \hat{B}} f_{\hat{A}_k}' (\hat{A}_1,\ldots, \hat{A}_n),	
	\end{align}
	where $f(\hat{A}_1,\ldots, \hat{A}_n)$ is defined by means of Eq. (\ref{Weyl_main_definition}).
\end{theorem}
\begin{proof}
	We introduce $\hat{A} \coloneqq i \sum_{q=1}^n \eta_q(\hat{A}_q - \xi_q)$ and $\hat{C} \coloneqq \commut{\hat{A}, \hat{B}} = i\sum_{q=1}^n \eta_q \hat{C}_q$;
	hence, $\commut{ \hat{A}, \hat{C} } = \commut{ \hat{B}, \hat{C} } = 0$. From the following identity:
	\begin{align}
		\commut{ \hat{A}_1 \cdots \hat{A}_n, \hat{B} }  &= 
			\sum_{k=1}^n \hat{A}_1 \cdots \hat{A}_{k-1} \commut{ \hat{A}_k, \hat{B} } \hat{A}_{k+1} \cdots \hat{A}_n,   
	\end{align}
	we obtain
	$
		\commut{ \hat{A}^k, \hat{B} } = k \hat{C} \hat{A}^{k-1}. 
	$
	It follows from Eq. (\ref{OperatorExp}) that
	\begin{align}
		\commut{ \exp(\hat{A}), \hat{B} } = \hat{C} \exp(\hat{A}) = \sum_{q=1}^n \hat{C}_q \frac{\partial}{\partial \hat{A}_q} \exp(\hat{A}).
	\end{align}
	Having substituted this equality into Eq. (\ref{Weyl_main_definition}), we finally reach Eq. (\ref{Weyl_MainTh_equality}).
\end{proof}

In particular, the canonical commutation relations $[ x, p ]= i \hbar$ imply the following identities
\begin{align}
 \label{xp-derivative}
 [ \hat{x} , G(\hat{x}, \hat{p}) ] &= 
 i \hbar \frac{\partial  G(\hat{x}, \hat{p}) }{\partial \hat{p}} \\
  [ \hat{p} , G(\hat{x}, \hat{p}) ] &= 
 -i \hbar \frac{\partial  G(\hat{x}, \hat{p})   }{\partial \hat{x}}.
\label{px-derivative}
\end{align}

\section{The Representation of Operators}
\label{algebra-representations}
The commutation relation
\begin{equation}
 [ \hat{x}^{\mu} , \hat{p}_{\nu} ] = - i \delta^{\mu}_{\,\,\,\,\nu} \hbar,
\end{equation}
is realized by the following representation
\begin{eqnarray}
  \hat{x}^{\mu} &\rightarrow & x^{\mu} \\
 \hat{p}_{\mu}  &\rightarrow& i \hbar \frac{\partial \,\,\,}{ \partial x^{\mu}}.
\end{eqnarray}
The momentum operator in non-relativistic mechanics $\hat{\mathbf{p}}$ is associated with $p_{\mu}$
according to $p^{\mu} = \{ E ,  \mathbf{p} \}$, which implies $p_{\mu} = \{ E ,  -\mathbf{p} \}$ and is consistent with
\begin{align}
  \hat{E} &\rightarrow  i \hbar \frac{\partial }{\partial t} \\ 
  \hat{\mathbf{p}}  &\rightarrow  - i \hbar \boldsymbol{\nabla}.
\end{align}
Similarly, the following commutator relations
\begin{eqnarray}
   [ \hat{X}^{\mu} , \hat{\Lambda}_{\nu} ] &=& - i \delta^{\mu}_{\,\,\,\,\nu},  \\
   {[} \hat{\Theta}^{\nu} ,   \hat{P}_{\mu}   {]} &=& - i \delta^{\mu}_{\,\,\,\,\nu}, 
\end{eqnarray}
can be expressed in the extended phase space as 
\begin{eqnarray}
  \hat{X}^{\mu} &\rightarrow & X^{\mu} ,\\
  \hat{\Lambda}_{\mu}  &\rightarrow& i \frac{\partial \,\,\,}{ \partial X^{\mu}},\\
   \hat{P}_{\mu} &\rightarrow & P_{\mu}, \\
  \hat{\Theta}^{\mu}  &\rightarrow& -i \frac{\partial \,\,\,}{ \partial P_{\mu}}.
\label{XP-representation}
\end{eqnarray}

An alternative representation is given by
\begin{eqnarray}
  \hat{X}^{\mu} &\rightarrow & X^{\mu} \\
  \hat{\Lambda}_{\mu}  &\rightarrow& i \frac{\partial \,\,\,}{ \partial X^{\mu}},\\
   \hat{P}_{\mu} &\rightarrow &  i \frac{\partial \,\,\,}{ \partial \Theta^{\mu}}, \\
  \hat{\Theta}^{\mu}  &\rightarrow& \Theta^{\mu}
\label{XTheta-representation},
\end{eqnarray}
which is defined as the \emph{double configuration space} wherein $\Theta^{\mu}$ acts
as a coordinate extension to $X^{\mu}$.

\section{The Classical  Spinor }
\label{classical-spinor-appendix}
The formulalism of classical spinors in classical relativistic mechanics naturally arises
from the standard formalism in the slash notation \cite{hestenes1973local,hestenes1999new,PhysRevA.45.4293}. 
Newton's equation in the Minkowski space for electromagnetic fields is
\begin{align}
m \frac{d u_{\mu}}{ds} &= e F_{\mu\nu} u^{\nu}.
\end{align}
The electromagnetic field  in this notation $F\!\!\!\!/$  is expressed through the 
gamma matrices as
\begin{equation}
 F\!\!\!\!/ = \frac{1}{2}\sigma^{\mu\nu} F_{\mu\nu},
\end{equation}
such that Newton's equation becomes 
\begin{align}
 \gamma^{0} \frac{d}{ds} u\!\!\!/ 
  &= \left \langle \gamma^{0} e F\!\!\!\!/ u\!\!\!/  \right \rangle_{H}.
\end{align}
This formula can be proven using the following identity
\begin{align}
  \langle \gamma^0 \sigma^{\mu\nu} \gamma_{\xi}  \rangle_H &= 
  \gamma^0( \gamma^{\mu}\delta^{\nu}_{\xi} - \gamma^{\nu} \delta^{\mu}_{\xi}  ),
\end{align}
with  $\langle A \rangle_H = \frac{1}{2}(A + A^\dagger)$, such that
\begin{align}
 \left \langle \gamma^{0} e F\!\!\!\!/ u\!\!\!/  \right \rangle_{H}
 =
 \frac{1}{2}\left \langle \gamma^{0} \sigma^{\mu\nu} \gamma_{\xi}  \right \rangle_{H} eF_{\mu\nu} u^{\xi}
 = 
  \frac{1}{2}\gamma^0( \gamma^{\mu}\delta^{\nu}_{\xi} - \gamma^{\nu} \delta^{\mu}_{\xi}  ) eF_{\mu\nu} u^{\xi}
 = \gamma^0 \gamma^{\mu} e F_{\mu\xi} u ^{\xi}.
\end{align}

The Lorentz rotor $L$ is defined as the double cover of the Lorentz group and 
can be parametrized in terms of three purely spatial rotation angles $\theta^k$ and three 
space-time rapidity variables $\eta_{k}$, such that  

\begin{align}
 L = B R = \exp{\left(\frac{1}{2} \eta_{k} \gamma^{0}\gamma^{k}  \right)}  
 \exp{\left( \frac{1}{4} \epsilon_{jkl} \theta^{j} \gamma^{k}\gamma^{l}  \right)}, 
\label{classical-Lorentz-rotor} 
\end{align}
where $B$ is a Hermitian matrix representing boosts and $R$ is a unitary matrix representing 
spatial rotations isomorphic to $SO(3)$.
The classical column spinor $\Psi$ is defined as the leftmost column of $L$ 
 as $ \Psi =   L |_{\text{leftmost column}}$ and the following relation
results \cite{hestenes1973local,hestenes1975observables}
\begin{align}
   \frac{d X^{\mu}}{ds} = \Psi^\dagger c \gamma^0\gamma^{\mu}\Psi,
\label{prop-vel-psi}
\end{align}
 where the proper velocity is expressed in the slash notation as
\begin{align}
 \frac{d}{ds} X\!\!\!\!/ \equiv   u\!\!\!/ = u^{\mu}\gamma_{\mu} = 
  \frac{1}{4}Tr( u\!\!\!/ \gamma^{\mu}   ) \gamma_\mu.
\end{align}

The Lorentz transformation law for the proper velocity is 
\begin{align}
   u\!\!\!/ \rightarrow u^\prime\!\!\!\!\!/ = L u\!\!\!/ L^{-1},
\end{align}
where the inverse Lorentz rotor obeys a simple formula $L^{-1} = \gamma^{0}L^{\dagger}\gamma^0$.
The proper velocity can be obtained as the active boost of the proper velocity of a 
particle initially at rest with proper velocity $ u\!\!\!/_{rest} = c \gamma^0$. This means that in general 
it is possible to find a Lorentz rotor such that
\begin{align}
  u\!\!\!/ = c L \gamma^0 L^{-1}. 
\end{align}
Using the decomposition $L = BR$, the proper velocity becomes 
\begin{align}
 u\!\!\!/ = c B^2 \gamma^0 = c \gamma^0 B^{-2}.
\end{align}
Thus, the boost from the proper velocity can be recovered as
\begin{align}
  B = \sqrt{ \frac{u\!\!\!/\gamma^0 }{c}  } = 
\frac{\frac{1}{c}  u\!\!\!/ \gamma^0 + \boldsymbol{1}}{\sqrt{2(\gamma+1)}},
\label{boost-formula}
\end{align}
with $\gamma = \frac{1}{4c}Tr(\gamma^0u\!\!\!/) = u^0/c$ as the relativistic length contraction factor.

Classical dynamics is described by the following equations 
\begin{align}
  \frac{d}{ds} X\!\!\!\!/ &=   u\!\!\!/ \\
  \gamma^{0} \frac{d}{ds} u\!\!\!/ ,
  &= \left \langle \gamma^{0} e F\!\!\!\!/ u\!\!\!/  \right \rangle_{H}, 
\end{align}
but the dynamics of the proper velocity can be assumed by the 
corresponding Lorentz rotor $L$. The general variation of the Lorentz rotor 
is expressed as \cite{baylis1999electrodynamics,Baylis1996book}
\begin{align}
  \frac{d L}{ds} = \frac{\Omega}{2} L, 
\end{align} 
where $\Omega$ lies in the space of the generators of the proper Lorentz group (double cover) 
$\sigma^{\mu\nu} = \frac{1}{2}[\gamma^{\mu},\gamma^{\nu}]$, containing both Hermitian and anti-Hermitian elements.
The derivative of the proper velocity in terms of $\Omega$ reads
\begin{align}
  \gamma^{0} \frac{d}{ds} u \!\!\!/ = \left \langle \gamma^0 \Omega u \!\!\!/  \right\rangle_{H},
\end{align}
which leads to the conclusion that it is \emph{sufficient} to identify $\Omega$ with 
the electromagnetic field. This fact can be used to establish the following 
alternative dynamical equations \cite{hestenes1999new,baylis1999electrodynamics,Baylis1996book}
\begin{align}
 \label{sde1}
 \frac{d}{ds} X\!\!\!\!/ &= c L L^\dagger \gamma^0, \\
  \frac{d L}{ds} &= \frac{e F\!\!\!\!/}{2} L,
 \label{sde2}
\end{align}
which are sometimes easier to solve than the standard dynamical equations of motion in
terms of the proper velocity \cite{baylis1999electrodynamics,PhysRevA.60.785,PhysRevA.45.4293}.

The projector  $\mathcal{P}$ can be defined as
\begin{align}
  \mathcal{P} \equiv \frac{1}{4}(\boldsymbol{1} + \gamma^0 )(\boldsymbol{1} + i \gamma^1\gamma^2 ) =
  diagonal\{1,0,0,0\},  
\end{align}
such that the following properties are verified
\begin{align}
   \gamma^0\mathcal{P}  &=  \mathcal{P} ,\\
   i\gamma^1\gamma^2\mathcal{P} &=  \mathcal{P}.
\label{projector-properties}
\end{align}

Equations (\ref{sde1}) and (\ref{sde2}) can be multiplied by $\mathcal{P}$ 
on the right and the proper velocity is expressed through the classical 
column spinor $\Psi$ 

\begin{align}
   \frac{d X^{\mu}}{ds} &= \Psi^\dagger c \gamma^0\gamma^{\mu}\Psi ,\\
  \frac{d \Psi}{ds} &= \frac{eF\!\!\!\!/}{2} \Psi.
\end{align}

\section{Solutions of the Dirac Equation for a Free Particle}
\label{Dirac-free-particle}
The  Dirac equation for the quantum Lorentz rotor in the absence
of external interactions is \cite{hestenes1973local,doran2003geometric,Baylis1996book}
\begin{align}
   \hbar \partial\!\!\!/ \mathcal{L}  \gamma^2\gamma^1 -  m c \mathcal{L}\gamma^0 = 0. 
\end{align}
The solutions can be sought in the form $\mathcal{L} = B R$, in analogy with the classical  
Lorentz rotor. The boost is expressed in terms of the
momentum by the following formulas
\begin{align}
 B_{+} &= \,\,\sqrt{ \frac{p\!\!\!/\gamma^0 }{mc}  }\,\, = 
\frac{ p\!\!\!/ \gamma^0 + mc  }{ \sqrt{2mc( p^0 + m c  )} }\,\,\,\,\,\,\,\, ;p^0 > mc > 0  ,\\
 B_{-} &=  \sqrt{- \frac{p\!\!\!/\gamma^0 }{mc}  } = 
\frac{ p\!\!\!/ \gamma^0 - mc  }{ \sqrt{-2mc( p^0 - m c  )} } \,\,; p^0 < - mc, 
\end{align}
where the restrictions are taken to ensure that the matrices are Hermitian. 
The first solution of the Dirac equation takes the form
\begin{align}
 \mathcal{L}_{+\uparrow} 
 = \frac{ p\!\!\!/ \gamma^0 + mc  }{ \sqrt{2mc( p^0 + m c  )} } 
\exp{( p\!\!\!/  \cdot x \!\!\!/ \, \gamma_1 \gamma_2/\hbar)},
\end{align}
which is consistent with the identification of a particle with 
positive energy $p^0>mc$. 

In order to proceed further, we resort to specific representations of the gamma matrices. 
The two most common representations for the gamma matrices are the Dirac  
and Weyl representations. The former one is
\begin{align}
  \gamma^0   &= \sigma_3 \otimes \mathbf{1}, \\
  \gamma^{j}  &= i \sigma_2 \otimes \sigma_j; 
\end{align}
whereas, the latter reads
\begin{align}
 \gamma^0 &=    \sigma_1 \otimes \mathbf{1}, \\
 \gamma^{j}&=  i\sigma_2 \otimes \sigma_j.
\end{align}
The  gamma matrices obey the following Clifford algebra in the Minkowski space
\begin{align}
   \frac{1}{2}(\gamma^{\mu}\gamma^{\nu} +\gamma^{\mu}\gamma^{\nu}) = g^{\mu \nu}, 
\end{align}
such that
\begin{align}
 {\gamma^0}^\dagger &= \gamma^0 ,\\
 (\gamma^k)^\dagger &= - (\gamma^k)^\dagger, \\
 \gamma^{\mu} &= \gamma_{\mu}^{-1}.
\end{align}

The corresponding column spinor in the Dirac representation is  
\begin{align}
  \psi_{+\uparrow} = \frac{1}{\sqrt{2mc( mc + p^0)}} 
          \begin{pmatrix}
          p^0 + mc  \\
          0        \\
          p^3 \\
          p^1 + i p^2    
          \end{pmatrix} \exp(-i  p\!\!\!/  \cdot x \!\!\!/ /  \hbar) = 
 \begin{pmatrix}
          1 \\
          0        \\
          0 \\
          0    
          \end{pmatrix} \exp(-i m c^2 \tau /  \hbar)
,
\end{align}
where $\tau$ is the proper time. The exponential indicates a fast rotation in the 
plane $\gamma_2\gamma_1$  with the \emph{zitterbewegung} frequency $\omega = \frac{2mc^2}{\hbar}$ in the reference 
frame attached to the particle \cite{hestenes2003mysteries,baylis2007broglie}.  
The momentum is calculated as the eigenvalue of the momentum operator defined 
as
\begin{align}
   \hat{\mathfrak{p}} \, \cdot = \hbar \partial\, \cdot \gamma^2 \gamma^1,
\end{align}
while the quantum current is calculated as 
\begin{align}
  J\!\!\!/ = \mathcal{L} {\mathcal{L}}^\dagger  \gamma^0.   
\end{align}
In the present case, the momentum is related to the 
current according with the following formula
\begin{align}
  p\!\!\!/_{+\uparrow} = mc {J\!\!\!/}_{+\uparrow} \gamma^0
\end{align}

The second independent solution is constructed by introducing an additional  $\pi$ 
spatial rotation on a plane perpendicular to $\gamma_1\gamma_2$, effectively reversing the orientation 
of the \emph{zitterbewegung} rotation as
\begin{align}
 \mathcal{L}_{+\downarrow} = \sqrt{ \frac{p\!\!\!/\gamma^0 }{mc}  } 
\exp( -\frac{\pi}{2} \gamma^2 \gamma^3  )
\exp{( p\!\!\!/  \cdot x \!\!\!/ \, \gamma_1 \gamma_2/\hbar)}. 
\end{align}
The column spinor representation for this case is
\begin{align}
  \psi_{+\downarrow} = \frac{i}{\sqrt{2mc( p^0 + mc)}} 
          \begin{pmatrix}
          0 \\
          p^0 + mc \\
          p^1 - i p^2\\
          -p^3    
          \end{pmatrix} \exp(-i  p\!\!\!/  \cdot x \!\!\!/  /\hbar ) =
 \begin{pmatrix}
          0 \\
          i        \\
          0 \\
          0    
          \end{pmatrix} \exp(-i m c^2 \tau /  \hbar).
\end{align}
The current and momentum for this second solution are exactly the same as for 
the first solution.

Additional independent solutions can be constructed by exploring the 
Hodge dual space through the multiplication of the volume element
  $  \gamma^0\gamma^1\gamma^2\gamma^3$ as  
\begin{align}
   \mathcal{L}_{-\uparrow} = \gamma^0\gamma^1\gamma^2\gamma^3
 \sqrt{- \frac{p\!\!\!/\gamma^0 }{mc}  } \exp{( p\!\!\!/  \cdot x \!\!\!/ \, \gamma^1 \gamma^2/\hbar)}. 
\end{align}   
This solution is consistent with a particle with negative energy $p^0<-mc$, 
as necessary requirement to satisfy the Dirac equation. The corresponding
column spinor representation is
\begin{align}
  \psi_{-\uparrow} = -\frac{i}{\sqrt{2mc( p^0 -m c  )}} 
          \begin{pmatrix}
          p^3 \\
          p^1 + i p^2\\
          p^0 - mc \\
          0      
          \end{pmatrix} \exp(-i  p\!\!\!/  \cdot x \!\!\!/  /\hbar ) = 
 \begin{pmatrix}
          0 \\
          0        \\
          i \\
          0    
          \end{pmatrix} \exp(i m c^2 \tau /  \hbar).
\end{align}
The current and momentum are calculated as
\begin{align}
   {J\!\!\!/}_{-\uparrow}  &= - \frac{p\!\!\!/\gamma^0 }{mc}, \\  
   {p\!\!\!/}_{-\uparrow}  &=  p\!\!\!/,
\end{align}
which indicates that they are anti-parallel. A serious difficulty is 
recognized by observing the  the zero component of the 
proper velocity
\begin{align}
  u^0 = c \frac{ dt}{d\tau}  = \frac{p^0}{mc} < 0, 
\end{align}
which is consistent with a particle moving to the past. This situation is
difficult or impossible to fit within classical mechanics and is inconsistent 
with the formulation of the classical spinor, which is explicitly defined with 
a positive zero component for the proper velocity. 

The fourth independent solution is found by applying an additional 
rotation on the previous one
\begin{align}
 \mathcal{L}_{-\downarrow} = \gamma^0\gamma^1\gamma^2\gamma^3\sqrt{ \frac{p\!\!\!/\gamma^0 }{-mc}  } 
\exp( -\frac{\pi}{2} \gamma^2 \gamma^3  )
\exp{( p\!\!\!/  \cdot x \!\!\!/ \, \gamma^1 \gamma^2/\hbar)}, 
\end{align}
with $p^0<-mc$. This solution suffers from the same problems as $\mathcal{L}_{-\downarrow}$.  The 
corresponding column spinor representation reads
\begin{align}
  \psi_{-\downarrow} = \frac{i}{\sqrt{2mc( p^0 -m c  )}} 
          \begin{pmatrix}
          p^1 - i p^2\\
          -p^3 \\
          0 \\
          p^0 - mc    
          \end{pmatrix} \exp(-i  p\!\!\!/  \cdot x \!\!\!/  /\hbar )
= 
 \begin{pmatrix}
          0 \\
          0        \\
          0 \\
          -1    
          \end{pmatrix} \exp(i m c^2 \tau /  \hbar).
\end{align}

\section{ODM: Quantum/Classical Spinless Salpeter Equations}
\label{SalpeterSec}
The general case for a spin zero particle is covered in the next section.  
The purpose of this section is to briefly introduce ODM and for this objective we 
present a simple illustrative example concerning the model for a spinless relativistic particle 
in one spatial dimension plus time (1+1), which in the quantum case leads to the spineless Salpeter equation.

Non-covariant relativistic classical mechanics can be expressed in terms of Newton's equation
\begin{align}
	\frac{dP}{dt} &= -U'(X),
\end{align}
with the momentum defined as  
\begin{align}
 \\
            p &= \frac{m  \frac{dX}{dt}  }{\sqrt{ 1 - \left( \frac{dX}{dt}/c\right)^2 }},
\end{align}
where $X$ are $P$ are the particle's coordinate and momentum. 
It can be shown that $  \frac{dX}{dt}  = \frac{P}{\sqrt{m^2 + (P/c)^2}}$, thus, the system of equations describing dynamics is
\begin{align}
		\frac{dP}{dt} &= -U'(X) , \\
 \frac{dX}{dt} &=  V'(P);  \qquad V(P) = \sqrt{ (cP)^2 + c^4 m^2 }.
\end{align}
According to ODM \cite{bondar2011hilbert} we can arrive to the same differential equations 
from a different set of postulates. 
The  \emph{first ODM element} we postulate is given by the law that governs the evolution of 
average values of an ensemble of independent particles
\begin{align}
 \label{firstODM-1}
		\frac{d\overline{P}}{dt} &= -\overline{U'(X)}  ,\\
 \frac{d\overline{X}}{dt} &=  \overline{V'(P)}.
\label{firstODM-2}
\end{align}
It must be emphasized that time-evolution for the averages are consistent 
for more than a specific kinematical description, i. e., classical kinematics.
In fact  the average values  do not even require the notion of 
localized trajectories and can be perfectly valid for quantum mechanics as well.

The \emph{second ODM element} involves the precise definition of the averaging, which
can be set in terms of linear operators and wavefunctions in the Hilbert space  
according to the postulates: i) The states of a system
are represented by normalized vectors $|\psi(t)\rangle$ of a complex Hilbert space, and the observables are given by
self-adjoint operators acting on this space;
 ii) The expectation value of a measurable $A$ at time t is $A(t) =  \langle \psi(t) | A | \psi(t) \rangle  $;
 iii) The probability that a measurement of an observable $\hat{A}$ at time t yields the value 
A is $| \langle A | \psi(t) \rangle  |^2$ where $ \hat{A} |A\rangle = A |A \rangle  $; iv) The state space of a composite system 
is the tensor product of the subsystems’ state spaces. These axioms are well-known in
quantum mechanics but it must be clear that taken alone, these axioms do not directly nor indirectly imply 
any restriction to exclusively describe quantum mechanics.
 Hence, the equations describing the averages in Eqs. (\ref{firstODM-1}-\ref{firstODM-2}) can be written 
in terms of a scalar wavefunction and the respective observable operators as
	\begin{align}
    \label{Ehrenfest-eqs-1}
		\frac{d}{dt} \bra{ \psi(t)} \hat{p} \ket{\psi(t)} &= \bra{\psi(t)} - U'(\hat{x}) \ket{\psi(t)}, \\
		\frac{d}{dt} \bra{ \psi(t)} \hat{x} \ket{\psi(t)} &= \bra{\psi(t)} V'(\hat{p}) \ket{\psi(t)},
   \label{Ehrenfest-eqs-2}
	\end{align} 
which resemble the Ehrenfest theorem.
The generator of motion $\hat{H}$, as a Hermitian operator, can be introduced through Stone's 
theorem as
	\begin{align}
		i\hbar \ket{d\psi(t)/dt} = \hat{H} \ket{\psi(t)}, 
	\end{align}
such that the time derivatives in the equations of expectation values (\ref{Ehrenfest-eqs-1}-\ref{Ehrenfest-eqs-2}) 
are  reduced to commutators
\begin{align}
  \bra{\psi(t)}  \frac{i}{\hbar} \commut{ \hat{H}, \hat{p} } \ket{\psi(t)}
 &=  \bra{\psi(t)} - U'(\hat{x}) \ket{\psi(t)} , \\
  \bra{\psi(t)}  \frac{i}{\hbar} \commut{ \hat{H}, \hat{x} } \ket{\psi(t)}
 &=  \bra{\psi(t)}  V'(\hat{p}) \ket{\psi(t)}.  
\end{align}
The expectation values can be dropped under general conditions because the identities must be
valid for all possible  initial states 
\begin{align}
\label{commutator-eqs-1}
    \frac{i}{\hbar} \commut{ \hat{H}, \hat{p} } 
 &=   - U'(\hat{x}),   \\
\label{commutator-eqs-2}
    \frac{i}{\hbar} \commut{ \hat{H}, \hat{x} } 
 &=   V'(\hat{p}) .  
\end{align}
The \emph{third ODM element} requires the specification of the algebra, 
which in the quantum case is achieved by demanding the canonical commutation 
relation between momentum and position
 \begin{align}
\commut{\hat{x}, \hat{p}} = i\hbar.
\end{align}
According to  Sec. \ref{Sec_Weyl_calculus}, $\commut{\hat{x}, \hat{p}} = i\hbar$  can be used to write 
Eqs. (\ref{commutator-eqs-1}-\ref{commutator-eqs-2}) as partial 
differential equations 
\begin{align}
   \frac{ \partial \hat{H}}{\partial \hat{x} } 
 &=   - U'(\hat{x})  ,\\
    \frac{ \partial \hat{H} }{ \partial \hat{p} } 
 &=   V'(\hat{p}) ,  
\end{align}
with a solution that corresponds to the quantum Salpeter Hamiltonian  
	\begin{align}
		\hat{H} = V(\hat{p}) + U(\hat{x}) = \sqrt{ (c\hat{p})^2 + c^4 m^2 } + U(\hat{x}),
	\end{align}
thus, the sought spinless Salpeter equation for the scalar wavefunction is 
\begin{align}
 \label{s-Salpeter}
   i \hbar \frac{\partial}{\partial t}\psi =  (\sqrt{ (c\hat{p})^2 + c^4 m^2 } + U(\hat{x}) )\psi
\end{align}

Having treated the quantum mechanical case, ODM can be used to find the corresponding 
classical description consistent with the evolution of averages (\ref{firstODM-1}-\ref{firstODM-2}) leading to  
the same Ehrenfest theorems  (\ref{Ehrenfest-eqs-1}-\ref{Ehrenfest-eqs-2}) but with 
a different labeling of the operators 
	\begin{align}
    \label{class-Ehrenfest-eqs-1}
		\frac{d}{dt} \bra{ \psi(t)} \hat{P} \ket{\psi(t)} &= \bra{\psi(t)} - U'(\hat{X}) \ket{\psi(t)} , \\
		\frac{d}{dt} \bra{ \psi(t)} \hat{X} \ket{\psi(t)} &= \bra{\psi(t)} V'(\hat{P}) \ket{\psi(t)}.
   \label{class-Ehrenfest-eqs-2}
	\end{align}  
In analogy with the quantum case, the generator of motion $\hat{K}$ is introduced 
by Stone's theorem 
\begin{align}
		i \ket{d\psi(t)/dt} = \hat{K} \ket{\psi(t)}. 
\end{align}
The time derivatives are expressed in terms of the generator of motion leading to
the following equations
\begin{align}
  \bra{\psi(t)}  i \commut{ \hat{H}, \hat{P} } \ket{\psi(t)}
 &=  \bra{\psi(t)} - U'(\hat{X}) \ket{\psi(t)} , \\
  \bra{\psi(t)}  i \commut{ \hat{H}, \hat{X} } \ket{\psi(t)}
 &=  \bra{\psi(t)}  V'(\hat{P}) \ket{\psi(t)}.  
\end{align}
The expectation values can be dropped according to the same arguments  
applied to the quantum case
\begin{align}
\label{class-commutator-eqs-1}
    i \commut{ \hat{K}, \hat{P} } 
 &=   - U'(\hat{X})  , \\
\label{class-commutator-eqs-2}
    i \commut{ \hat{K}, \hat{X} } 
 &=   V'(\hat{P}) .  
\end{align}
The \emph{third ODM element} complies with  classical mechanics by 
demanding the commutativity of the position and momentum operators
\begin{align}
 [ \hat{X} , \hat{P} ] = 0.
\end{align} 
However, assuming a generator of motion only depending
on the position and momentum operators  $K = K(\hat{X},\hat{P})$ leads to a contradiction 
with Eqs. (\ref{class-commutator-eqs-1}-\ref{class-commutator-eqs-2}) 
because it implies that all commutators on the left side are identically 
equal to zero. 
This problem can be avoided by extending the algebra with two additional 
operators $\hat{\Lambda}$ and $\hat{\Theta}$ obeying the algebra
\begin{align}
 {[} \hat{X}, \hat{\Lambda}  {]} &= i, \\
 {[} \hat{P}, \hat{\Theta}  {]} &= i, \\
  {[} \hat{\Lambda}, \hat{\Theta}  {]} &= 0.
\end{align}
According to Sec. \ref{Sec_Weyl_calculus}, this algebra 
of classical operators is used to write  Eqs. (\ref{class-commutator-eqs-1}-\ref{class-commutator-eqs-2})
in terms of partial derivatives as
\begin{align}
     \frac{ \partial \hat{K}}{ \partial \hat{\Theta} } 
 &=   - U'(\hat{X}),   \\
     \frac{ \partial \hat{K} }{\partial \hat{\Lambda} } 
 &=   V'(\hat{P}),  
\end{align}
with the following solution
\begin{align}
 \hat{K} = V'(\hat{P}) \hat{\Lambda} - U'(\hat{X}) \hat{\Theta} + f(\hat{X},\hat{P}  ), 
\end{align}
where $ f(\hat{X},\hat{P}  )$ is an arbitrary function of the position and momentum.
An explicit representation of the classical algebra is constructed in terms of differential
operators in the phase space
\begin{align}
   \hat{X} &= X,\\
   \hat{P} &= P,\\
   \hat{\Lambda} &= -i \frac{\partial}{\partial X}, \\
   \hat{\Theta} &= -i \frac{\partial}{\partial P}.
\end{align}
Thus, the classical wavefunction's $\Psi$ equation of motion is
\begin{align}
 i \frac{\partial }{\partial t} \Psi = 
 \left( -i V'(P)   \frac{\partial}{\partial X} +i  U'(X)  \frac{\partial}{\partial P}  + f(\hat{X},\hat{P}  )   \right)\Psi . 
\label{classical-Salpeter} 
\end{align}
Furthermore, the density is defined as $\rho = \psi \psi^*$, leading to the following equation of motion
\begin{align}
 \frac{\partial}{ \partial t } \rho =
  \left( - V'(P)   \frac{\partial}{\partial X} +  U'(X)  \frac{\partial}{\partial P}      \right)\rho ,  
\end{align}
which can be written in terms of Poisson brackets for a classical Hamiltonian $H_{c} = V(P) +U(X)$ as
\begin{align}
 \frac{\partial }{\partial t} \rho = \Poisson{H_c,\rho}
\end{align} 
In summary, we have obtained the quantum spinless Salpeter equation (\ref{s-Salpeter}) and the 
corresponding classical counterpart (\ref{classical-Salpeter}), wherein the the state is represented 
by quantum/classical wavefunctions in the Hilbert space. The solution of the 
classical dynamical equation (\ref{classical-Salpeter}) is a complex scalar wavefunction but its 
density can be interpreted in terms of standard classical mechanics. Therefore, 
 a formulation in the Hilbert space is not a necessary condition to ascertain quantumness.

\section{ODM: Spinless Quantum/classical Relativistic Mechanics}
\label{Klein-Gordon-section}
The previous section considered  a model for one dimensional relativistic quantum/classical mechanics 
but mostly concerned to illustrate  ODM than to provide a proper description of the topic.
 In this section we  provide a more rigurous treatment of relativistic mechanics in 
manifestly covariant 
form with electromagnetic field interaction described by
\begin{align}
m \frac{d u_{\mu}}{ds} &= e F_{\mu\nu} u^{\nu},
\end{align}
where the proper velocity components $u_{\mu}$ are defined in terms of the derivative of the space-time 
coordinates $X^\mu$ 
\begin{align}
  u^{\mu}  =   \frac{d X^{\mu}}{ds}.
\end{align}
An equivalent formulation can be developed in terms of the position-momentum 
canonical coordinates (Sec. \ref{standard-classical-dynamics})
\begin{align}
 \label{proper-velocity-4}
 \frac{d X^{\mu}}{ds} &= \frac{P^{\mu}-eA^{\mu}}{m} ,\\
\frac{d P_{\mu}}{ds} &=  \frac{e}{m} (\partial_{\mu} A_{\nu})( P^{\nu} -e A^{\nu}  ).
\label{canonical-force-4}
\end{align}
The components $X^{\mu}$ and $P_{\mu}$ define the \emph{extended phase space}, which is the natural 
generalization of the phase space in non-relativistic mechanics excluding $X^0$ and $P_0$. 
The classical equations of motion (\ref{proper-velocity-4}) 
and (\ref{canonical-force-4}) can be used as the
 \emph{first ODM element} such that  
\begin{align}
\label{first-ODM-KG-1}
 \frac{d \overline{X^{\mu}} }{ds} &= \overline{ \frac{P^{\mu}-eA^{\mu}}{m} } , \\
\frac{d \overline{P_{\mu}}}{ds} &=  \frac{e}{m} \overline{ (\partial_{\mu} A_{\nu})( P^{\nu} -e A^{\nu}  ) }.
\label{first-ODM-KG-2}
\end{align}
The \emph{second ODM element} is invoked and the average values are expressed
in terms of a scalar classical wavefunction $\ket{\Psi}$ in the Hilbert space and the corresponding
observable operators in the form of Ehrenfest theorems as
\begin{align}
 \frac{d \langle \Psi | \hat{X}^{\mu} | \Psi \rangle }{ds} &= 
 \langle \Psi | \frac{\hat{P}^{\mu}-e \hat{A}^{\mu}}{m} | \Psi \rangle  , \\
 \frac{ d \langle \Psi | \hat{P}_{\mu} | \Psi \rangle }{ds} &= 
 \langle \Psi |
 \frac{1}{2}\Acommut{ 
  \frac{e}{m} \partial_{\mu} \hat{A}_{\nu}
,
   \hat{P}^{\nu} -e\hat{A}^{\nu}
}
| \Psi \rangle,
\end{align}
where the anticommutator $\{\cdot,\cdot\}$ was applied in order to enforce Hermiticity.
Stones's theorem introduces the generator of motion $\boldsymbol{\mathcal{K}}$
\begin{align}
  \frac{d | \Psi \rangle}{ds}  = i \boldsymbol{\mathcal{K}} |\Psi \rangle,
\end{align}
which is utilized to obtain the following system of equations
\begin{align}
i[ \hat{X}^{\mu} , G ] &= \frac{\hat{P}^{\mu}-e\hat{A}^{\mu}}{m}  , \\
 i[ \hat{P}_{\mu} , G] &=\frac{e}{2m}
  \acommut{
  \partial_{\mu} \hat{A}_{\nu}
  ,
  \hat{P}^{\nu}-e\hat{A}^{\nu}
  }.
\end{align}
\emph{Third ODM element}. The generator of motion is assumed to depend on classical 
operators  $\boldsymbol{\mathcal{K}} =\boldsymbol{\mathcal{K}}( \hat{X} , \hat{P} , \hat{\Lambda} , \hat{\Theta}  )  $ obeying the following classical algebra 
\begin{align}
 \label{commutation-classical-1}
 [ \hat{X}^{\mu} , \hat{P}_{\nu} ] &= 0 , \\
  \label{commutation-classical-2}
 {[} \hat{X}^{\mu} , \hat{\Lambda}_{\nu}  {]} &= -i \delta^{\mu}_{\,\,\,\,\nu} , \\
  \label{commutation-classical-3}
 {[}\hat{P}_{\mu} , \hat{\Theta}^{\nu}  {]} &= -i \delta^{\nu}_{\,\,\,\,\mu} , \\
 {[}\hat{\Lambda}_{\mu} , \hat{\Theta}^{\nu}{]} &= 0.  
 \label{commutation-classical-4}
\end{align}
Applying the properties of the classical algebra according to 
Sec. \ref{Sec_Weyl_calculus},
the generator of motion must satisfy the equations
\begin{align}
  \frac{\partial \boldsymbol{\mathcal{K}}  }{\partial \Lambda_{\mu}} 
&= \frac{\hat{P}^{\mu}-e\hat{A}^{\mu}}{m} , \\
   \frac{\partial \boldsymbol{\mathcal{K}} }{\partial \Theta^{\mu}}     
&=\frac{1}{2m}\Acommut{ 
 e\partial_{\mu} \hat{A}_{\nu} ,
  \hat{P}^{\nu}-e\hat{A}^{\nu}
},
\end{align}
whose solution is the relativistic Koopman-von Neumann generator
\begin{align}
\boldsymbol{\mathcal{K}} = \frac{1}{2m}
\Acommut{
\hat{P}^{\mu}-e\hat{A}^{\mu}
,
 \hat{\Lambda}_{\mu} + e \partial_{\nu} \hat{A}_{\mu} \hat{\Theta}^{\nu}  
}
+ f(\hat{X},\hat{P}), \label{classicalK}
\end{align}
with $f$ as an arbitrary function of the classical operators $\hat{X}^{\mu}$
and $\hat{P}_{\mu}$. The anticommutator can be expanded and upon applying the commutation 
relations (\ref{commutation-classical-1}-\ref{commutation-classical-4}) leads to
\begin{align}
   K = \frac{1}{m}
 ( \hat{P}^{\mu} - e\hat{A}^{\mu})(  \hat{\Lambda}_{\mu} + e \partial_{\nu} \hat{A}_{\mu} \hat{\Theta}^{\nu} )
 + i \frac{e}{m} \partial_{\mu}A^{\mu} + f(\hat{X},\hat{P})
,
\end{align}
where  $ f(\hat{X},\hat{P})$ can be chosen to cancel $ i \frac{e}{m} \partial_{\mu}A^{\mu} $. 
The relativistic Koopman-von Neumann equation 
can now be stated in terms of the  extended phase-space representation  (\ref{XP-representation})  
as
\begin{align}
 \left[
 \frac{i}{m}
 ( P^{\nu} - e A^{\nu}  )
 \left( \frac{\partial}{\partial X^{\nu}} + e \partial_{\alpha} A_{\nu} \frac{\partial}{\partial P_{\alpha}}   \right) 
+  f(X,P)
 \right]
\Psi = 0.
\label{rel-KN}
\end{align}
Moreover, demanding $f(X,P)$ to be real, the density $\rho=\Psi\Psi^*$ is seen to
satisfy the relativistic Liouville equation
\begin{align}
 \frac{1}{m}
 ( P^{\nu} - e A^{\nu}  )
 \left( \frac{\partial}{\partial X^{\nu}} + e \partial_{\alpha} A_{\nu} \frac{\partial}{\partial P_{\alpha}}   \right) 
\rho = 0.
\label{rel-rel-liouville}
\end{align}

Having expressed classical dynamics in the extended  phase space, there is still 
another possibility to be explored: classical mechanics in the 
configuration-tangent space wherein the velocity replaces the momentum 
as state variable. The relativistic dynamical equations expressed in terms of the 
position and velocity are employed as the \emph{first ODM element} 
\begin{align}
  \frac{d \overline{X^{\mu}}}{ds}  &= \overline{ u^{\mu} } , \\
m \frac{d \overline{ u_{\mu} }}{ds} &= e \overline{  F_{\mu\nu} u^{\nu} }. 
\end{align} 
In analogy with the ODM procedure followed in Sec. \ref{SalpeterSec}, the generator of motion $G$ can be introduced through Stone's theorem such that
\begin{align}
  i[ \hat{X}^{\mu}   , G ] &= \hat{u}^{\mu} , \\
 i[ \hat{u}_{\mu}   , G ] &= e F_{\mu \nu} \hat{u}^{\nu}. 
\end{align}
Assuming the dependence  $G = G(\hat{X}^{\mu}, \hat{u}_{\mu}, \hat{\Lambda}_{\mu},
  \hat{\Theta}^{\mu} )$, with the following underlying algebra
\begin{align}
 [ \hat{X}^{\mu} , \hat{u}_{\nu} ] &= 0 , \\
 {[} \hat{X}^{\mu} , \hat{\Lambda}_{\nu}  {]} &= -i \delta^{\mu}_{\,\,\,\,\nu} , \\
 {[}\hat{u}_{\mu} , \hat{\Theta}^{\nu}  {]} &= -i \delta^{\nu}_{\,\,\,\,\mu} , \\
 {[}\hat{\Lambda}_{\mu} , \hat{\Theta}^{\nu}{]} &= 0,  
 \label{commutation-classical-u}
\end{align}
the generator of motion satisfy
\begin{align}
\frac{\partial G }{ \partial \hat{\Lambda}_{\mu}}   &= \hat{u}^{\mu} , \\
  \frac{\partial G }{ \partial \hat{\Theta}^{\mu}}  &= e F_{\mu \nu} \hat{u}^{\nu},
\end{align}
with the final solution
\begin{equation}
  G = \hat{u}^{\mu} \hat{\Lambda}_{\mu}  
  + e F_{\mu \nu} \hat{u}^{\nu} \hat{\Theta}^{\mu} + 
  f(\hat{x},\hat{u}).
\label{preVlasov}
\end{equation}
This generator can be used to define the equation of motion for a scalar 
wavefunction $\Psi$. The corresponding scalar density can be defined as $\rho = \Psi\Psi^*$,
leading to 
\begin{align}
  u^{\mu } \partial_{\mu} \rho  + e F_{\mu \nu} u^{\nu} \frac{\partial \rho }{ \partial u_{\mu} } = 0,
\end{align}
which is the relativistic Vlasov equation  \cite{Liboff2003} upon 
imposing a coupling with the Maxwell equations in the context of plasma physics.     

Having achieved a consistent classical formulation in terms of a classical scalar wavefunction, we 
can now concentrate on the respective quantum formulation for a particle without spin.
The first and second ODM elements are  exactly the same as in the classical case  leading to the same expectation values
\begin{align}
 \frac{d \langle \psi | \hat{x}^{\mu} | \psi \rangle }{ds} &= 
 \langle \psi | \frac{\hat{p}^{\mu}-e \hat{A}^{\mu}}{m} | \psi \rangle  , \\
 \frac{ d \langle \psi | \hat{p}_{\mu} | \psi \rangle }{ds} &= 
 \langle \psi |
 \frac{1}{2}\Acommut{ 
  \frac{e}{m} \partial_{\mu} \hat{A}_{\nu}
,
   \hat{p}^{\nu} -e\hat{A}^{\nu}
}
| \psi \rangle,
\end{align}
The first difference appears at the time of introducing the generator of motion  $\mathcal{K}$ through
Stones's theorem 
\begin{align}
 \hbar \frac{d | \psi \rangle}{ds}  = i \mathcal{K} |\psi \rangle,
\end{align}
where the Planck contant is employed to obtain the following
system of equations
\begin{align}
 \frac{i}{\hbar}[ \hat{x}^{\mu} , \mathcal{K} ] &= \frac{\hat{p}^{\mu}-e\hat{A}^{\mu}}{m}  , \\
 \frac{i}{\hbar}[ \hat{p}_{\mu} , \mathcal{K} ] &= 
 \frac{1}{2m}\Acommut{
 e\partial_{\mu} A_{\nu}
 ,
 \hat{p}^{\nu}-e \hat{A}^{\nu}
 }
. 
\label{commG}
\end{align}
In the quantum case, it is enough to assume the dependence on position and 
momentum operators
 $ \mathcal{K} = \mathcal{K}(\hat{x}^{\mu},\hat{p}_{\mu})$, obeying the quantum commutation relations as \emph{third ODM element}
\begin{align}
 [ \hat{x}^{\mu} , \hat{p}_{\nu} ] = - i \hbar \delta^{\mu}_{\,\,\,\,\nu}.
\end{align}
The system of differential equations for the quantum generator is
\begin{align}
   \frac{\partial \mathcal{K}}{ \partial \hat{p}_{\mu}  }  &= \frac{\hat{p}^{\mu} - e \hat{A}^{\mu}}{m}\\
  \frac{\partial \mathcal{K}}{ \partial \hat{x}^{\mu}} &=
   -\frac{1}{2m}\Acommut{
   e\partial_{\mu} \hat{A}_{\nu}
  ,
   \hat{p}^{\nu}- e \hat{A}^{\nu}
  } .
\label{partialG}
\end{align}
The solution of this equation with the integration constant $\frac{mc^2}{2}$ is the 
Klein-Gordon generator with electromagnetic interaction
\begin{align}
  \mathcal{K} =  \frac{1}{2m}(\hat{p}^{\mu}-e\hat{A}^{\mu})(\hat{p}_{\mu}-e\hat{A}_{\mu}) - \frac{mc^2}{2}, 
\label{GK-generator}
\end{align}
which has the same form as the classical Hamiltonian (see Sec. \ref{standard-classical-dynamics}) with operators 
instead of scalar variables.

\section{ODM: The Pauli Equation}
\label{PauliEqSec}
The purpose of this section is to develop non-relativistic spinorial  mechanics through ODM.
The spinorial quantum description is given by the Pauli equation, while the classical one does 
not exist.   
The following fundamental averages' evolution of the position and momentum  
was postulated in \cite{bondar2011hilbert} as \emph{first ODM element} 
\begin{align}
\label{basic-averages-1}
    \frac{d}{dt} \overline{{X}_k } 
&= \frac{1}{m} \overline{  {P}_k  } , 
\\
    \frac{d}{dt} \overline{ {P}_k } 
&=         - \overline{\partial_k U({\bf X})}, 
\label{basic-averages-2}
\end{align}
and used to deduce both quantum and classical mechanics upon specifying the underlying algebra.
The present situation involves magnetic fields and spin degrees 
of freedom $S_j$. The interaction with the magnetic field is naturally prescribed by the 
minimal coupling $P_k \rightarrow  P_k - e A_k $. The averages' evolution of the spin degrees 
of freedom must be compatible with the time-evolution of a magnetic dipole, so we postulate
\begin{align}
   \frac{d}{dt} \overline{ {S}_i }  &= -\frac{e}{m} \sum_{j,k=1}^3 \epsilon_{ijk} \overline{ {S}_k  B_j ({\bf x}) }.
\end{align}
Considering the outcome from the Stern-Gerlach experiment, we postulate a
 spin-orbit coupling given by $ \frac{e}{m} \sum_{j=1}^3 \overline{{S}_j \partial_k B_j({\bf x})} $. Placing all the elements together,
the fundamental law for the averages as \emph{first ODM element} read
\begin{align}
    \frac{d}{dt} \overline{{X}_k } 
&= \frac{1}{m} \overline{ \left[ {P}_k - e A_k({\bf X}) \right] } ,
\\
    \frac{d}{dt} \overline{ {P}_k } 
&= \sum_{j=1}^3 \frac{e}{m} \overline{ \left[ {P}_j - e A_j({\bf X}) \right] \partial_k A_j ({\bf X})}
            - e \overline{\partial_k V({\bf X})} +
 \frac{e}{m} \sum_{j=1}^3 \overline{{S}_j \partial_k B_j({\bf X})} , \\
    \frac{d}{dt} \overline{ {S}_i } 
&= -\frac{e}{m} \sum_{j,k=1}^3 \epsilon_{ijk} \overline{ {S}_k  B_j ({\bf X}) }. 
\end{align}

\emph{Second ODM element}. The explicit definition of the averages, can be expressed in terms
 of expectation values of a wave function in the spinorial Hilbert space $\ket{\Psi(t)}$, 
with the variables replaced by their corresponding operators. 
Hereinafter, the dependencies on $\hat{\bf x}$ and $t$ will be omitted for the sake of simplicity 
such that
\begin{align}
    \frac{d}{dt} \langle \Psi | \hat{x}_k  | \Psi \rangle
&= \frac{1}{m}  \langle \Psi | \hat{p}_k - e \hat{A}_k | \Psi \rangle , \\
     \frac{d}{dt}  \langle \Psi | \hat{p}_k   | \Psi \rangle
&=  \langle \Psi | \sum_{j=1}^3 \frac{e}{m}  \left[ \hat{p}_j - e \hat{A}_j \right] \partial_k \hat{A}_j
            - e \partial_k V   + \frac{e}{m} \sum_{j=1}^3 \hat{S}_j \partial_k B_j | \Psi \rangle , \\
    \frac{d}{dt}   \langle \Psi | \hat{S}_i  | \Psi \rangle 
&= - \langle \Psi | \frac{e}{m} \sum_{j,k=1}^3 \epsilon_{ijk}  \hat{S}_k  B_j| \Psi \rangle  . 
\end{align}
Introducing the Hamiltonian as  $i\hbar \ket{d\Psi(t)/dt} = \hat{H} \ket{\Psi(t)}$, and dropping the expectation 
values, we obtain the following commutator equations
\begin{align}
\label{PauliCommutatorEqs}
    \frac{i}{\hbar} \commut{ \hat{H}, \hat{x}_k } &= 
  \frac{1}{m} \left( \hat{p}_k - e\hat{A}_k \right) , \nonumber\\
    \frac{i}{\hbar} \commut{ \hat{H}, \hat{p}_k } &= 
  \frac{e}{m} \sum_{j=1}^3 \left( \hat{p}_j - e \hat{A}_j \right) \partial_k \hat{A}_j - e \partial_k V +
   \frac{e}{m} \sum_{j=1}^3 \hat{S}_j \partial_k B_j , \\
    \frac{i}{\hbar} \commut{ \hat{H}, \hat{S}_i }  &= -\frac{e}{m} \sum_{j,k=1}^3 \epsilon_{ijk} B_j \hat{S}_k \nonumber.
\end{align}

The {\emph{third ODM element}} is introduced through the 
following quantum commutation relations
\begin{align}\label{PauliEqAlgebra}
    \commut{\hat{x}_k, \hat{p}_j} = i\hbar\delta_{kj}, \qquad \commut{ \hat{S}_i, \hat{S}_j } = i\hbar \sum_{k=1}^3 \epsilon_{ijk} \hat{S}_k,
\end{align}
where the spin is set to obey the standard commutation relationship for the angular momentum.
In order to proceed further, the generator of dynamics $\hat{H}$ is sought in the following form
\begin{align}\label{AnzatsForPauliEq}
    \hat{H} = F_0 (\hat{\bf x}, \hat{\bf p}) + \sum_{j=1}^3 F_j (\hat{\bf x}, \hat{\bf p}) \hat{S}_j.
\end{align}
Upon upon imposing the quantum commutation relations 
and equating the coefficients for each spin component, 
Eq. (\ref{PauliCommutatorEqs}) result in
\begin{align}
    \frac{\partial F_0}{\partial p_k} &= \frac{1}{m} \left( p_k - e \hat{A}_k \right), 
   \\ \frac{\partial F_j}{\partial p_k} &= 0 , \\
    \frac{\partial F_0}{\partial x_k} &=
   -\frac{e}{m} \sum_{j=1}^3 \left( \hat{p}_j - e \hat{A}_j \right) \partial_k \hat{A}_j + e \partial_k V , \\
    \frac{\partial F_j}{\partial x_k} &=  -\frac{e}{m}  \partial_k B_j , \\
   F_j &=  -\frac{e}{m} B_j. 
\end{align}
This system of differential equation can be solved to finally obtain the 
quantum Pauli Hamiltonian 
\begin{align}
    \hat{H} =  \frac{1}{2m} \sum_{k=1}^3 \left( \hat{p}_k -  e \hat{A}_k \right)^2 + 
 e V - \frac{e}{m} \sum_{j=1}^3 \hat{S}_j B_j.
\end{align}

The corresponding classical 
description in the spinorial Hilbert space can be attempted by replacing the quantum 
algebra by a consistent classical algebra where the position and momentum operators 
commute $\commut{ \hat{X}_k, \hat{P}_j } = 0$. This classical algebra must be extended with 
auxiliary operators $ \hat{\Lambda}_{j}$, $ \hat{\Theta}_{j} $ for the same reasons explained in the previous sections.  
In the present case, the classical algebra can be chosen as
\begin{align}
\label{ClassicalPauliAlgebra}
    \commut{ \hat{X}_k, \hat{P}_j } &= 0, \qquad \commut{ \hat{X}_k, \hat{\Lambda}_{l} } = i\delta_{k,l},
\\  \commut{ \hat{P}_k, \hat{\Theta}_{l} } &= i\delta_{k,l},
    \qquad \commut{ \hat{S}_i, \hat{S}_j } = i\hbar \sum_{k=1}^3 \epsilon_{ijk} \hat{S}_k.
\end{align}
Substituting the Ansatz
\begin{align}
    \hat{H} = G_0 (\hat{\bf X}, \hat{\bf P}, \hat{\bf \Lambda},  \hat{\bf \Theta} )
 + \sum_{j=1}^3 G_j (\hat{\bf X}, \hat{\bf P}, \hat{\bf \Theta},  \hat{\bf \Lambda}) \hat{S}_j
\end{align} 
into Eq. (\ref{PauliCommutatorEqs}), one obtains a system of partial 
differential equations 
\begin{align}
    \frac{\partial G_0}{\partial \Lambda_{k}} &= \frac{\hbar}{m} \left( P_k - e\hat{A}_k \right) , \\
    \frac{\partial G_j}{\partial \Lambda_{k}} &= 0 , \\
    \frac{\partial G_0}{\partial \Theta_{k}} &= \frac{e\hbar}{m} \sum_{j=1}^3 
  \left( \hat{P}_j - e \hat{A}_j \right) \partial_k \hat{A}_j - e\hbar \partial_k V,  \notag\\
     \frac{\partial G_j}{\partial \Theta_{k}} &=  \frac{e\hbar}{m} \partial_k B_j \\
     G_j &= -\frac{e}{m} B_j.
\end{align}
The last two equalities are in general not compatible unless the magnetic 
field is constant. This type of magnetic field is widely used
to model specific systems but any real magnetic field cannot be truly 
constant because Maxwell's equations demand the magnetic field lines 
to form closed curves. In other words, the commutator
equations (\ref{PauliCommutatorEqs}) do not have a solution within 
the classical algebra (\ref{ClassicalPauliAlgebra}). 
Thus, we conclude that there is no classical spin in the non-relativistic framework.
\end{widetext}

\bibliography{bib-relativity}

\end{document}